\documentclass[a4paper, 12pt]{article}
\usepackage{setspace}
\usepackage[margin=1.2in]{geometry}
\usepackage{amsmath,amsfonts,amsthm} 
\usepackage{graphicx,cite,times}
\usepackage{enumitem,version}
\usepackage{pdfpages,hyperref}

\newtheorem{theorem}{Theorem}[section]

\newtheorem{lemma}[theorem]{Lemma}

\newtheorem{example}{Example}
\newtheorem{remark}{Remark}
\newtheorem{definition}{Definition}[section]

\numberwithin{equation}{section}
\Large
\begin{document}
	\title{Cryptography using generalized Fibonacci matrices with Affine-Hill cipher}
	\author{Kalika Prasad$^{1}$\footnote{E-mail: klkaprsd@gmail.com, \indent ORCiD ID: https://orcid.org/0000-0002-3653-5854/}, Hrishikesh Mahato$^{2}$\footnote{E-mail: hrishikesh.mahato@gmail.com}
		\\
		\normalsize{$^{1,2}$Department of Mathematics, Central University of Jharkhand, India, 835205}   }      
	\date{\today}
	\maketitle
	\noindent\rule{15cm}{.15pt}
	\begin{abstract}
		In this article, we have proposed a public key cryptography using Affine-Hill cipher with generalised Fibonacci matrix(called multinacci matrix). Also proposed a key establishment(exchange of key matrix $K=Q_{\lambda}^{k}$ of order $\lambda\times\lambda$ for encryption-decryption) scheme with the help of multinacci sequences under prime modulo. In this scheme, instead of exchanging key matrix, we need to exchange only pair of numbers $(\lambda, k)$, which reduces the time complexity as well as space complexity and comes with a large key-space.
	\end{abstract}
	\noindent\rule{15cm}{.1pt}
		
\section{Introduction}
	In classical cryptography, 
	Hill cipher{\cite{stallings2017cryptography,stinson2005cryptography,mao2003modern}} is one of the polygraphic substitution ciphers which is based on residue system and linear algebra which was developed by the mathematician Lester Hill in 1929.
	
	M.K. Viswanath, et.al\cite{viswanath2015public} proposed the concept of public key cryptography using Hill's Cipher. They developed public key cryptography with Hill’s cipher system using rectangular matrix and for inverse of key matrix, they used Moore–Penrose Inverse (Pseudo Inverse) method. Later, P. Sundarayya \& G.V. Prasad\cite{sundarayya2019public} worked on same paper\cite{viswanath2015public}, and they proposed the method which increase the security of the above system by involving two or more digital signatures. To increase security of  Hill Cipher, Thilaka and Rajalakshmi\cite{thilaka2005extension} extended the concept of Hill Cipher by enciphering of m-length string to n-length string($ n \geq m $) using affine transformation and polynomial transformation. While In \cite{gupta2007cryptanalysis}, Indivar Gupta, et.al shows that these extension of Hill Cipher are prone to cryptanalytic attack and suggested that modifications in Hill Cipher do not make it significantly stronger. 
	
	In this paper, we develop a public key cryptosystem using Affine-Hill Cipher with generalised Fibonacci(multinacci) matrix with large power k, i.e  $Q_{\lambda}^{k}$ as key. Our method is quite robust and can be implemented easily.
	\\This paper is organized as follow. Following the introduction, the basic concept of Hill's Cipher \& Affine Cipher is outlined in Section 2. In Section 3, we discuss about Fibonacci \& Multinacci sequences and development of matrix corresponding to these sequences. In Section 4, discussion of proposed method for generating key matrix and inverse of key matrix with an numerical example. Finally in Section 5, we presented crypt-analytic strength and concluding remarks.
	
\section{Preliminaries}
	\subsection*{Hill's Cipher}
	Hill cipher is the first polygraphic cipher\footnote{Polygraphic cipher: In encryption, more than one letter can be encrypted at a time. }. The Hill's encryption scheme takes $ n $ successive plaintext letters and substitutes them for $ n $ ciphertext letters.
	Here, we use matrix representation P for plaintext, K for key matrix and C for ciphertext(encrypted matrix), where P, K, and Q are given as
	\begin{equation}
	P =
	\begin{bmatrix}	
		P_{1}& P_{2} & ... & P_{m}
	\end{bmatrix},
	K=
	\begin{bmatrix}	
		K_{11}& K_{12} & ... & K_{1n} \\ 		
		K_{21} & K_{22} & ... & K_{2n} \\ 		
		\vdots & \vdots & \ddots & \vdots \\ 		
		K_{n1} & K_{n2} & ... & K_{nn} 	
	\end{bmatrix}
	~~\text{and}~~
	C=
	\begin{bmatrix}	
		C_{1}& C_{2} & ... & C_{m}
	\end{bmatrix}
	\end{equation}	
	where $P_{i}'s$ and $C_{i}'s$ are block matrix of size $1\times n$.
	Thus, Hill cipher is described as:
	\\For encryption
	\begin{center}
		$Enc(P)$:\hspace{1cm} $C_{i} \equiv P_{i}K \pmod p$
	\end{center}
	\vspace{.4in}
	For decryption
	\begin{center}
		$Dec(C):$\hspace{.7cm} $P_{i} \equiv  K^{-1} C_{i}  \pmod p$.
	\end{center}
	where $p$ is prime and $gcd(det(K),p)=1$.
	
	\subsection*{Affine Cipher}
	One of the special case of substitution cipher is \textbf{Affine Cipher}\cite{stinson2005cryptography}, which is described as 
	\begin{eqnarray}
		Enc(x): y = (ax+b) \pmod{26}, \text{where} ~~a,b \in Z_{26}
	\end{eqnarray}
	These functions are called \textbf{affine functions} and in the affine cipher, we restrict the encryption function to affine function.
	\\For decryption, we have to solve the above function for $x$, i.e 
	\begin{eqnarray}
		y &\equiv& (ax+b) \pmod {26} \nonumber\\
		i.e~~ x &\equiv& a^{-1}(y-b) \pmod {26}, \text{where}~ a,b,x,y \in Z_{26}
	\end{eqnarray}
	where $gcd(a,26)=1$, which confirms the existence of $a^{-1}$.
			
	\subsection*{Affine-Hill Cipher}
	Affine-Hill Cipher is polygraphic block cipher, which extends the concept of the Hill Cipher as described above. Affine-Hill Cipher's encryption-decryption technique is given as:
	\begin{eqnarray}\label{Affine-Hill}
		Enc(P):\hspace{.7cm} C_{i} &\equiv& (P_{i}K + B) \pmod p \\
		Dec(C):\hspace{.7cm} P_{i} &\equiv& (C_{i} - B)K^{-1} \pmod p		
	\end{eqnarray}
	Where $P_{i}, C_{i}$ and $B$ are $1\times n$ matrices, $K$ is $n\times n$ key matrix, $p$ is prime greater than number of different characters used in plaintext and $Enc(P)$ \& $Dec(C)$ represents encryption techniques and decryption techniques respectively.
	\subsection{Key exchange Algorithm(ElGamal Technique)}\label{ElGamal}
	In 1984, T. Elgamal\cite{stallings2017cryptography,stinson2005cryptography,paar2009understanding} proposed a public-key scheme based on discrete logarithms\cite{elgamal1985public} closely related to Deffie-Hellman technique. In Elgamal technique, the global elements are the chosen prime $p$ and selected primitive root of $p$.
	
	Let $p$ be a prime. Now chose a private key D such that $1 < D < \phi(p)$, then further select a primitive root\footnote{The exponent of $\alpha$ modulo $p$ is the least positive integer $k$ such that $\alpha^k\equiv 1 \pmod p$. If the exponent of $\alpha$ modulo $p$ is $\phi{(p)}$ then $\alpha$ is said to be a primitive root of $p$.} of $p$, say $\alpha$.
	Further assign $E_{1}=\alpha$ and $E_{2}= E_{1}^{D}\pmod p$, then make ($p, E_{1}, E_{2}$) as public key and keep $D$ as secret key. Suppose entity Alice and Bob want to exchange key.
	
	\subsubsection{Key exchanging}\label{EGkeyexchange}
	Using above public key($p, E_{1}, E_{2}$), Alice generates k as
	\\(i).~ Choose a random integer $e$ such that $1 < e < \phi(p)$, then
	\\(ii). Compute signature $k = E_{1}^{e} \pmod p$
	\\(iii). Compute secret key $\lambda  = E_{2}^{e} \pmod p$
	\\(iv). Thus, Alice who have access to Bob's public key can encrypt message with their secret key $\lambda $ and send $(k, C)$.
	\subsubsection*{Key Recover by Bob:}
	After receiving ($k, C)$ from Alice, Bob recovers secret key $\lambda $ using their secret key $ D $ as:
	\begin{eqnarray}\label{recoverkey}
		\lambda  &=& c^{D}\pmod p \nonumber\\
		  &\equiv& (E_{1}^{e})^{D} \pmod p \nonumber\\
		  &\equiv& (E_{1}^{D})^{e}\pmod p\nonumber\\
		  &\equiv& (E_{2})^{e}\pmod p
	\end{eqnarray}
	Thus, here Bob received the secret key $\lambda $ securely and using this secret key $\lambda$, Bob will decrypt the ciphertext $C$, and get original plaintext $P$.
	
\section{Fibonacci Sequence and Fibonacci Q$_\lambda$ Matrix}
	The Fibonacci sequence\cite{johnson2008fibonacci,grimaldi2012fibonacci,wiki:xxx} is the sequence of integers $f_{n}$, defined by the recurrence relation
	\begin{eqnarray}\label{Fib2Seq}
		f_{k+2} = f_{k} + f_{k+1}, ~~k\geq 0,\text{ with initial values}~f_{0}=0,f_{1}=1
	\end{eqnarray}
	Fibonacci-Q$_\lambda$ matrix was introduced by Brenner\cite{grimaldi2012fibonacci}, later King were enumerated it's basic properties. In 1985, Honsberger\cite{honsberger1985mathematical} showed that the Fibonacci Q-matrix is a square matrix of order 2 of the form
	$ Q_{2} = 
	\begin{bmatrix}
		f_{2} & f_{1} \\ 
		f_{1} & f_{0} \\
	\end{bmatrix}
	= 
	\begin{bmatrix}
		1 & 1 \\ 
		1 & 0\\
	\end{bmatrix}$.
	\\
	And, it can be observed that $ Q_{2}^{n} = 
	\begin{bmatrix}	
		f_{n+1} & f_{n} \\ 
		f_{n} & f_{n-1} \\
	\end{bmatrix}$ with
	$ Q_{2} = 
	\begin{bmatrix}
	1 & 1 \\ 
	1 & 0 \\
	\end{bmatrix}$.
	where, Fibonacci $Q_2$-matrix\cite{gould1981history} has been developed using Fibonacci sequence and, it's inverse is given by
	\begin{eqnarray}\label{Q1-negMatrix}
		(Q_{2}^{n})^{-1} = (Q_{2})^{-n} = 
		\begin{bmatrix}		
		f_{-n+1} & f_{-n} \\ 		
		f_{-n} & f_{-n-1} \\
		\end{bmatrix}
	\end{eqnarray}
	Where $f_{-n}$ represent negative Fibonacci numbers. As we know that, Fibonacci sequence can be extended in negative which is given by $f_{-n} = (-1)^{n+1}f_{n}, ~~~n=1,2,3,..$
	\\And the determinant of $ Q_{2}^{n} $ matrix is given by 
	\begin{eqnarray}
		det(Q_{2}^{n}) = f_{n+1}f_{n-1} - f_{n}^{2} 
	\end{eqnarray}
	\textbf{Cassini's Identity:} Cassini's Identity\cite{koshy2019fibonacci} is mathematical identity for the Fibonacci numbers. It was discovered in 1680 by J.D. Cassini by observing the pattern of Fibonacci sequence. Cassini's Identity states that for the $ n^{th} $-Fibonacci number
	\begin{eqnarray}\label{cassini}
		f_{n-1}f_{n+1}-f_{n}^2 &=& (-1)^{n}, \hspace{1cm}\text{where} \hspace{.1cm} n\in \mathbb{Z}\nonumber\\
		\text{Thus,}~~~~~ det(Q_{2}^{n}) &=& (-1)^{n}
	\end{eqnarray}
	$\implies Q_{2}^{n}$ is non-singular matrix, for all $n$. 
	\\In continuation to above Fibonacci matrix, Tribonacci matrix(for $\lambda = 3$), $Q_{3}^{n}$ is defined as
	\begin{eqnarray}
		 Q_{3}^{n} = 
		\begin{bmatrix}
			f_{n+2} & f_{n+1} + f_{n} & f_{n+1}\\ 
			f_{n+1} & f_{n} + f_{n-1} & f_{n} \\
			f_{n} & f_{n-1} + f_{n-2} & f_{n-1} \\
		\end{bmatrix}
		~~with~~~
		Q_{3} = 
		\begin{bmatrix}
			1 & 1 & 1\\	
			1 & 0 & 0\\
			0 & 1 & 0\\
		\end{bmatrix}
	\end{eqnarray} And inverse of $Q_{3}^{n}$ is given by 
	\begin{eqnarray}
		Q_{3}^{-n}  = 
		\begin{bmatrix}	
			f_{-n+2} & f_{-n+1} + f_{-n} & f_{-n}\\ 		
			f_{-n+1} & f_{-n} + f_{-n-1} & f_{-n-1} \\ 			
			f_{-n} & f_{-n-1} + f_{-n-2} & f_{-n-2}\\	
		\end{bmatrix}
		~~with~~~Q_{3}^{-1} =
		\begin{bmatrix}	
			0 & 1 & 0\\ 		
			0 & 0 & 1 \\ 			
			1 & -1 & -1\\	
		\end{bmatrix}
	\end{eqnarray}
	\\where tribonacci sequences is given by recurrence relation 
	$f_{k+3}=f_{k}+f_{k+1}+f_{k+2} $, $k \in \mathbb{Z}$ with initial values $f_{0}=f_{1}=0$ \& $f_{2}=1$.
	\subsection{Generalized Fibonacci Sequence and Fibonacci matrix}
	The $\lambda^{th}$ order generalized Fibonacci sequence is given by the following $\lambda^{th}$ order recurrence relation:
	\begin{eqnarray}\label{genrec}
		f_{k+\lambda}= f_{k}+ f_{k+1}+...+ f_{k+\lambda-1} \hspace{2cm}k\geq0 , \text{$ \lambda \in \mathbb{Z^+} $}
	\end{eqnarray}
	with initial values $f_{0}=f_{1}=f_{2}=....= f_{\lambda-2}=0, f_{\lambda-1}=1$.
	\\Consider the corresponding multinacci $Q_{\lambda}$-matrix of order $\lambda$, given by 
	\begin{center}
	$Q_{\lambda}=
	\begin{bmatrix}	
		1 & 1 & ... & 1 & 1\\	
		1 & 0 & ... & 0 & 0\\
		0 & 1 & ... & 0 & 0\\
		\vdots & \vdots & \ddots & \vdots \\
		0 & 0 & ... & 1 & 0\\
	\end{bmatrix}_{\lambda\times\lambda}
	=
	\begin{bmatrix}
		f_{\lambda} & f_{\lambda-1} +...+f_{0} & f_{\lambda-1}+...+f_{1} & ... & f_{\lambda-1} \\
		
		f_{\lambda-1} & f_{\lambda-2} +...+f_{-1} & f_{\lambda-2}+...+f_{0} & ... & f_{\lambda-2} \\
		
		\vdots & \vdots & \vdots & \ddots & \vdots \\
		
		f_{2} & f_{1} +...+f_{-\lambda} & f_{1}+...+f_{1-\lambda} & ... & f_{1} \\
		f_{1} & f_{0}+ ...+f_{-\lambda-1} & f_{0}+...+f_{-\lambda} & ... & f_{0} \\
	\end{bmatrix}$ 
	\end{center}
	And using mathematical induction, it can be observed that 	
	\begin{eqnarray}\label{GenFeb}
	\begin{bmatrix}	
		1 & 1 & 1 & ... & 1 & 1\\	
		1 & 0 & 0 & ... & 0 & 0\\
		0 & 1 & 0 & ... & 0 & 0\\
		\vdots & \vdots & \vdots & \ddots & \vdots \\
		0 & 0 & 0 & ... & 1 & 0\\
	\end{bmatrix}^{k}
	&=&
	\begin{bmatrix}
		f_{k+\lambda-1} & f_{k+\lambda-2}+...+f_{k-1} & f_{k+\lambda-2}+...+f_{k} & ... & f_{k+\lambda-2} \\
		
		f_{k+\lambda-2} & f_{k+\lambda-3} +...+f_{k-2} & f_{k+\lambda-3}+...+f_{k-1} & ... & f_{k+\lambda-3} \\
		
		\vdots & \vdots & \vdots & \ddots & \vdots \\
		
		f_{k+1} & f_{k}+ ...+f_{k-\lambda-1} & f_{k}+...+f_{k-\lambda} & ... & f_{k} \\
		f_{k} & f_{k-1}+ ...+f_{k-\lambda-2} & f_{k-1}+...+f_{k-\lambda-1} & ... & f_{k-1}\nonumber\\
	\end{bmatrix}\\
	&=& Q_{\lambda}^{k}
	\end{eqnarray}
	
	\begin{remark}
		Let $Q_{\lambda}^{k}$ is generalized Fibonacci matrix of order $\lambda\times\lambda$, then
		\begin{eqnarray}
		det(Q_{\lambda}) &=& (-1)^{\lambda-1}\hspace{1cm} (here~~Q_{\lambda}= Q_{\lambda}^{1})\nonumber\\
		\text{Thus,}\hspace{1cm} det(Q_{\lambda}^{k}) &=& [(-1)^{\lambda-1}]^k\nonumber\\
					&=& (-1)^{(\lambda-1)k}
		\end{eqnarray}
	\end{remark}
	We also observed that
	$Q_{\lambda}^k.Q_{\lambda}^1 = Q_{\lambda}^{k+1}$ As\\
	$Q_{\lambda}^k.Q_{\lambda}^1 =$\\
	$\begin{bmatrix}
		f_{k+\lambda-1} & f_{k+\lambda-2}+f_{k+\lambda-3}+...+f_{k-1} & f_{k+\lambda-2}+...+f_{k} & ... & f_{k+\lambda-2} \\
		
		f_{k+\lambda-2} & f_{k+\lambda-3}+f_{k+\lambda-4}+...+f_{k-2} & f_{k+\lambda-3}+...+f_{k-1} & ... & f_{k+\lambda-3} \\
		
		\vdots & \vdots & \vdots & \ddots & \vdots \\
		
		f_{k+1} & f_{k}+f_{k-1}+...+f_{k-\lambda-1} & f_{k}+...+f_{k-\lambda} & ... & f_{k} \\
		f_{k} & f_{k-1}+f_{k-2}+...+f_{k-\lambda-2} & f_{k-1}+...+f_{k-\lambda-1} & ... & f_{k-1} \\
	\end{bmatrix} . 
	\begin{bmatrix}	
	1 & 1 & 1 & ... & 1 & 1\\	
	1 & 0 & 0 & ... & 0 & 0\\
	0 & 1 & 0 & ... & 0 & 0\\
	\vdots & \vdots & \vdots & \ddots & \vdots \\
	0 & 0 & 0 & ... & 1 & 0\\
	\end{bmatrix}$
	\vspace{.4cm}
	\\=~~
	$\begin{bmatrix}
	f_{k+\lambda-1}+f_{k+\lambda-2}+...+f_{k-1} & f_{k+\lambda-1}+f_{k+\lambda-2}+...+f_{k} & f_{k+\lambda-1}+...+f_{k+1} & ... & f_{k+\lambda-1} \\
	
	f_{k+\lambda-2}+f_{k+\lambda-3}+...+f_{k-2} & f_{k+\lambda-2}+f_{k+\lambda-3}+...+f_{k-1} & f_{k+\lambda-2}+...+f_{k} & ... & f_{k+\lambda-2} \\
	
	\vdots & \vdots & \vdots & \ddots & \vdots \\
	
	f_{k+1}+f_{k}+...+f_{k-\lambda-1} & f_{k+1}+f_{k}+...+f_{k-\lambda} & f_{k+1}+...+f_{k-\lambda+1} & ... & f_{k+1} \\
	f_{k}+f_{k-1}+...+f_{k-\lambda-2} & f_{k}+f_{k-1}+...+f_{k-\lambda-1} & f_{k}+...+f_{k-\lambda} & ... & f_{k} \\
	\end{bmatrix}$
	\vspace{.1cm}
	\\entries of first column can be reduces using recurrence relations(\ref{genrec}), which gives 
	\\=~~
	$\begin{bmatrix}
	f_{k+\lambda} & f_{k+\lambda-1}+f_{k+\lambda-2}+...+f_{k} & f_{k+\lambda-1}+...+f_{k+1} & ... & f_{k+\lambda-1}\\	
	
	f_{k+\lambda-1} & f_{k+\lambda-2}+f_{k+\lambda-3}+...+f_{k-1} & f_{k+\lambda-2}+...+f_{k} & ... & f_{k+\lambda-2} \\

	\vdots & \vdots & \vdots & \ddots & \vdots \\
	
	f_{k+2} & f_{k+1}+f_{k}+...+f_{k-\lambda} & f_{k+1}+...+f_{k-\lambda+1} & ... & f_{k+1} \\
	f_{k+1} &f_{k}+f_{k-1}+...+f_{k-\lambda-1} & f_{k}+...+f_{k-\lambda} &...& f_{k}\\
	\end{bmatrix}  = Q_{\lambda}^{k+1}$
	\\\\Further, the recurrence relation (\ref{genrec}) can be re-write as 
	\begin{eqnarray}\label{negrecrel}
		f_{k}= f_{k+\lambda}-(f_{k+1}+...+ f_{k+\lambda-1}) \hspace{2cm} for~ k\leq -1
	\end{eqnarray}
	Or, equivalently
	\begin{eqnarray}\label{neg-k}
	f_{-k}= f_{-k+\lambda}-(f_{-k+1}+...+ f_{-k+\lambda-1}) \hspace{2cm} for~ k\geq 1
	\end{eqnarray}
	with the same initial values given in equation(\ref{genrec}). Equation(\ref{negrecrel}) generalize the $\lambda^{th}$ order negative multinacci sequences.
	
	\begin{lemma}
		Let p is prime and K is generalized Fibonacci matrix, then 
		\begin{eqnarray}
		det(K)\pmod p = det(K\pmod p)
		\end{eqnarray}
	\end{lemma}
	\begin{theorem}[Existence of Inverse of multinacci matrix] 
		Let $\lambda\in\mathbb{Z^{+}}$, then for every integer $k\in \mathbb{Z}$, the inverse of multinacci matrices $Q_{\lambda}^k$ is given by $Q_{\lambda}^{-k}$ as defined in equation(\ref{GenFeb}).
	\end{theorem}
	\begin{proof}
	We shall prove existence by mathematical induction on k. Since, by the definition of $Q_{\lambda}^k(\ref{GenFeb})$, we have
	\begin{center}
		$Q_{\lambda}^{-1} =
		\begin{bmatrix}	
		0 & 1 & 0 & ... & 0 & 0\\	
		0 & 0 & 1 & ... & 0 & 0\\
		0 & 0 & 0 & ... & 0 & 0\\
		\vdots & \vdots & \vdots & \ddots & \vdots & \vdots \\
		0 & 0 & 0 & ... & 0 & 1\\
		1 & -1 & -1 & ... & -1 & -1\\
		\end{bmatrix}_{\lambda \times\lambda}$
	\end{center}
	and
	\begin{eqnarray}
	Q_{\lambda}^{-k} =
	\begin{bmatrix}
		f_{-k+\lambda-1} & f_{-k+\lambda-2}+f_{-k+\lambda-3}+...+f_{-k-1} & ... & f_{-k+\lambda-2} \\
			
		f_{-k+\lambda-2} & f_{-k+\lambda-3}+f_{-k+\lambda-4}+...+f_{-k-2} &... & f_{-k+\lambda-3} \\
			
		\vdots & \vdots & \ddots & \vdots \\
			
		f_{-k+1} & f_{-k}+f_{-k-1}+...+f_{-k-\lambda-1} &... & f_{-k} \\
		f_{-k} & f_{-k-1}+f_{-k-2}+...+f_{-k-\lambda-2} &... & f_{-k-1} \\
	\end{bmatrix}
	\end{eqnarray}
	Now, for $k=1$
	\begin{eqnarray}\label{truefor=1}
	Q_{\lambda}^{1}.Q_{\lambda}^{-1}
	&=& 
	\begin{bmatrix}	
	1 & 1 & 1 & ... & 1\\	
	1 & 0 & 0 & ... & 0 \\
	0 & 1 & 0 & ... & 0 \\
	\vdots & \vdots & \vdots & \ddots & \vdots \\
	0 & 0 & 0 & ... & 1 \\
	\end{bmatrix}~.~ 
	\begin{bmatrix}	
	0 & 1 & 0 & ... & 0\\	
	0 & 0 & 1 & ... & 0 \\
	0 & 0 & 0 & ... & 0 \\
	\vdots & \vdots & \vdots & \ddots & \vdots \\
	1 & -1 & -1 & ... & -1 \\
	\end{bmatrix}
	=
	\begin{bmatrix}	
	1 & 0 & 0 & ... & 0\\	
	0 & 1 & 0 & ... & 0 \\
	0 & 0 & 1 & ... & 0 \\
	\vdots & \vdots & \vdots & \ddots & \vdots \\
	0 & 0 & 0 & ... & 1 \\
	\end{bmatrix}_{\lambda \times \lambda}\nonumber\\
	&=& I_{\lambda}~~\hspace{4cm}
	\end{eqnarray}
	Since, we have $Q_{\lambda}^{-1}.Q_{\lambda}^{-k} =$
	\begin{center}
	$\begin{bmatrix}	
		0 & 1 & 0 & ... & 0\\	
		0 & 0 & 1 & ... & 0\\
		\vdots & \vdots & \vdots & \ddots & \vdots \\
		0 & 0 & 0 & ... & 1\\
		1 & -1 & -1 & ... &-1\\
	\end{bmatrix} . 		
	\begin{bmatrix}
		f_{-k+\lambda-1} & f_{-k+\lambda-2}+f_{-k+\lambda-3}+...+f_{-k-1} & ... & f_{-k+\lambda-2} \\
			
		f_{-k+\lambda-2} & f_{-k+\lambda-3}+f_{-k+\lambda-4}+...+f_{-k-2} &... & f_{-k+\lambda-3} \\
		
		\vdots & \vdots & \ddots & \vdots \\
		
		f_{-k+1} & f_{-k}+f_{-k-1}+...+f_{-k-\lambda-1} &... & f_{-k} \\
		f_{-k} & f_{-k-1}+f_{-k-2}+...+f_{-k-\lambda-2} &... & f_{-k-1} \\
	\end{bmatrix}$\vspace{.4cm}
	=
	$\begin{bmatrix}
		f_{-k+\lambda-2} & f_{-k+\lambda-3}+f_{-k+\lambda-4}+...+f_{-k-2} &... & f_{-k+\lambda-3} \\
		
		f_{-k+\lambda-3} & f_{-k+\lambda-4}+f_{-k+\lambda-5}+...+f_{-k-3} & ... & f_{-k+\lambda-4}\\	
		
		f_{-k+\lambda-4} & f_{-k+\lambda-5}+f_{-k+\lambda-6}+...+f_{-k-4} & ... & f_{-k+\lambda-5} \\		
		\vdots & \vdots & \ddots & \vdots \\	
		
		f_{-k} & f_{-k-1}+f_{-k-2}+...+f_{-k-\lambda-2} &... & f_{-k-1} \\
		f_{-k-1} & f_{-k-2}+f_{-k-3}+...+f_{-k-\lambda-3} &... & f_{-k-2} \\
	\end{bmatrix} = Q_{\lambda}^{-(k+1)}$
	\end{center}
	Now, assume that result holds for $k = m$, i.e
	\begin{eqnarray}\label{truefor=k}
		Q_{\lambda}^{m}.Q_{\lambda}^{-m} = I_{\lambda}
	\end{eqnarray}
	Thus, for $k = m+1$, we have,
	\begin{eqnarray}
		Q_{\lambda}^{(m+1)}.Q_{\lambda}^{-(m+1)} 
		&=& Q_{\lambda}^{m}.Q_{\lambda}^{1}Q_{\lambda}^{-1}.Q_{\lambda}^{-m} \nonumber\\
		&=& Q_{\lambda}^{m}.I_{\lambda}.Q_{\lambda}^{-m} ~~~~~~~~~~\text{Using equation(\ref{truefor=1})} \nonumber\\
		&=& Q_{\lambda}^{m}.Q_{\lambda}^{-m} \hspace{1.5cm}\text{Using equation(\ref{truefor=k})} \nonumber\\
		&=& I_{\lambda}
	\end{eqnarray} Hence, proved
	\end{proof}
\section{Proposed Encryption Schemes}
	\subsection*{Calculation of $ n^{th} $ Fibonacci term}
	Since, the $ n^{th} $ Fibonacci number is defined as the sum of the $(n-1)^{th} $ and $ (n-2)^{th} $ term. So to calculate the $ n^{th} $ Fibonacci number, for instance, we need to compute all the $ n-1 $ values before it first - quite a task, even with a calculator!
	
	There is explicit formula for recurrence relation of order two by using which we can be directly calculate the $n^{th}$ term of sequence, which reduces the time and space complexity $\mathcal{O}(n^2)$ to $\mathcal{O}(1)$, where $\mathcal{O}$ represents $Big~O -notation$. It is given using solutions of following characteristic equation, corresponding to recurrence relation(\ref{Fib2Seq})
	\begin{eqnarray}\label{ch-eqn}
		\alpha^{2}-\alpha - 1 = 0
	\end{eqnarray}
	and explicit formula for $n^{th}$ term of Fibonacci sequence is
	\begin{eqnarray}
		F_{n} &=& \frac{1}{\sqrt{5}}\left[\left(\frac{1+\sqrt{5}}{2}\right)^{n} - \left(\frac{1-\sqrt{5}}{2}\right)^{n}\right] \nonumber
	\end{eqnarray}

	\subsubsection*{Calculation of $ n^{th}$ term for tribonacci numbers}	
	Keeping in mind the complexity for construction of key matrix for modified hill's cipher, we may establish a formula for direct calculation of $n^{th}$ term of tribonacci sequence using the generalization of explicit formula for higher order. 
	Similar to Fibonacci sequence, the $ n^{th}$ term of tribonacci sequence is given using solutions of following characteristic equation of degree three
	\begin{eqnarray}\label{4.7}
		\alpha^{3}-\alpha^{2}-\alpha - 1 = 0
	\end{eqnarray}
	On solving equation(\ref{4.7}), we get roots $\alpha_{1}, \alpha_{2}$ \text{and} $\alpha_{3}$, which are \vspace{.2cm} \\
	$\alpha_{1} = \dfrac{(19+3\sqrt{33})^{\frac{2}{3}}+(19+3\sqrt{33})^{\frac{1}{3}}+4}{(19+3\sqrt{33})^{\frac{1}{3}}}$,\\
	$\alpha_{2} = \dfrac{(19+3\sqrt{33})^{\frac{2}{3}}(-1+i\sqrt{3})+2(19+3\sqrt{33})^{\frac{1}{3}}-4(1+i\sqrt{3})}{6(19+3\sqrt{33})^{\frac{1}{3}}}$, and\\
	$\alpha_{3} = \dfrac{(19+3\sqrt{33})^{\frac{2}{3}}(1-i\sqrt{3})-2(19+3\sqrt{33})^{\frac{1}{3}}+4(1-i\sqrt{3})}{6(19+3\sqrt{33})^{\frac{1}{3}}}$\vspace{.3cm}.
	\\
	And $n^{th}$ term of tribonacci is given by
	\begin{eqnarray}
		F_{n}=C_{1}(\alpha_{1})^{n} + C_{2}(\alpha_{2})^{n} + C_{3}(\alpha_{3})^{n}, \hspace{.4cm} n\in\mathbb{Z}
	\end{eqnarray}
	where, \\
	$C_{1}= {\frac {\frac{1}{2} \left( \sqrt [3]{19+3\, \sqrt{3} \sqrt{11}}+2 \right)  \left( 19+3\, \sqrt{3} \sqrt{11} \right) \\
				\mbox{}}{ \sqrt{33}\left( 19+3\, \sqrt{33}\right)^{\frac{2}{3}}+2\, \sqrt{11}\sqrt [3]{19+3\, \sqrt{33}} \sqrt{3}\\
				\mbox{}+23\, \sqrt{33} + 9\, \left( 19+3\, \sqrt{33} \right) ^{\frac{2}{3}}+18\,\sqrt[3]{19+3\,\sqrt{33}}\\\mbox{}+135}}$\vspace{.2cm}\\
	$C_{2} = {\frac {-i/12 \left( -9\,i \sqrt{11}+4\,i\sqrt [3]{19+3\, \sqrt{33}} \sqrt{3}\\
			\mbox{}-9\, \sqrt{33}-19\,i \sqrt{3}-12\,\sqrt [3]{19+3\, \sqrt{33}}\\
			\mbox{}-57 \right)  \sqrt{3}\sqrt [3]{19+3\, \sqrt{33}}}{99+19\, \sqrt{3} \sqrt{11}\\
			\mbox{}}}$\vspace{.2cm}\\
	$C_{3} = \frac {-i/12 \sqrt{3}\sqrt [3]{19+3\, \sqrt{33}}\\
		\mbox{} \left( -9\,i \sqrt{11}+4\,i\sqrt [3]{19+3\, \sqrt{33}} \sqrt{3}\\
		\mbox{}+9\, \sqrt{33}-19\,i \sqrt{3}+12\,\sqrt [3]{19+3\,\sqrt{33}}\\
		\mbox{}+57 \right) }{99+19\, \sqrt{33}}\vspace{.2cm}$
	\begin{remark}
		For higher order Fibonacci sequence, we can construct corresponding difference equation but solution of higher order difference equation(greater than 3) is still complicated.
	\end{remark}
	\subsection{Key Exchange Scheme}
	Suppose we have public key $pk(p, E_{1},E_{2})$, where component of public key $E_{1} ~and ~E_{2}$ is created by Bob(receiver) with the help of their private key $D$. Now, using this public key, the secret key $\lambda$ will be calculated(See, \ref{ElGamal}). After generating secret key $\lambda$, key matrix $K$ will be given as 
	\subsubsection{Algorithm}
	\textbf{Encryption Algorithm:}
	\begin{enumerate}
		\item Alice chooses secret number $e$, such that $1 < e < \phi(p)$
		\item \textbf{Signature:} $k \leftarrow E_{1}^{e}\pmod{p}$
		\item \textbf{Secret key:} $\lambda \leftarrow E_{2}^{e}\pmod{p}$
		\item \textbf{Key Matrix:} $K \leftarrow Q_{\lambda}^{k}$, where $Q_{\lambda}^{k}$ is multinacci matrix of order $\lambda\times \lambda$ (See, \ref{GenFeb}).
		\item  \textbf{Encryption:} $Enc(P): $\hspace{.76cm} $C_{i}\leftarrow (P_{i}K + B) \pmod{p}$
		\item  transmit $(k, C)$ to Bob
	\end{enumerate}
    \textbf{Decryption Algorithm:} Bob, after receiving $(k, C)$
	\begin{enumerate}
		\item \textbf{Secret key:} $\lambda\leftarrow k^{D}\pmod{p}$(see, \ref{recoverkey}), where $D$ is Bob's secret key. 
		\item \textbf{Key Matrix:} $K \leftarrow Q_{\lambda}^{k}$.
		\item \textbf{Decryption:} $Dec(C) :$ \hspace{1cm} $P_{i} \leftarrow(C_{i} - B)K^{-1}\pmod{p}$
	\end{enumerate}
	\subsection{Numerical Example}
	 Suppose Alice want to send message to Bob, then first she calculate key matrix K using above proposed technique and then encrypt plaintext P with key matrix K.
	\begin{example}[Public key calculation]\label{eg1}
		Assume that $p=37$ and Bob's private key is $D=13$. Further Bob choose primitive root of $p$, say $\alpha=5$.
	\end{example}
	 Bob assign $E_{1}=\alpha=5$ and compute $E_{2}=E_{1}^{D}\pmod{p} = 5^{13}\pmod{37} = 13$
	 \\Thus, Bob's public key $ pk(p, E_{1},E_{2})$ is $pk(37,5,13) $ and secret key is $D=13 $.
	\begin{example}[Encryption-Decryption]
	Suppose plaintext be \textbf{HELLO2019}, public key is $ pk(37,5,13) $ and shifting vector B is [31,13,19].
	\end{example}
	\begin{proof}[Solution]
	Here plaintext $P=\textbf{HELLO2019}$. So for encryption\\
	First, choose $ e $ such that $1 < e < \phi(p)$, let $ e=22 $.\\
	Calculating signature, $k=E_{1}^{e} = 5^{22}\pmod{37}\equiv 4$.
	\\and secret key $\lambda = E_{2}^{e} = 13^{22}\pmod{37} \equiv 3$.
	\\Now, constructing the key matrix K using above data according to generalized Fibonacci matrix\vspace{.2cm} $Q_{\lambda}$(\ref{GenFeb}), which is given by \\
	$K = Q_{\lambda}^{k} = Q_{3}^{4} = $
	$\begin{bmatrix}\label{Q2}
		f_{6} & f_{5}+f_{4} & f_{5} \\			
		f_{5} & f_{4}+f_{3} & f_{4} \\		
		f_{4} & f_{3}+f_{2} & f_{3} \\			
	\end{bmatrix}
	= \begin{bmatrix}
		7 & 6 & 4 \\	
		4 & 3 & 2 \\	
		2 & 2 & 1 \\	
	\end{bmatrix}$
	and $K\pmod{37} = 
	\begin{bmatrix}
		7 & 6 & 4 \\	
		4 & 3 & 2 \\	
		2 & 2 & 1 \\
	\end{bmatrix}$
	\\\\where generalized Fibonacci sequence for $\lambda = 3$(Tribonacci) is given as\vspace{.3cm}
	\\\resizebox{\linewidth}{!}
	{
	\begin{tabular}{|c|c|c|c|c|c|c|c|c|c|c|c|c|c|c|c|c|c|c|c|c|c|c|c|c|c|c|}
			\hline 
		Index &... & -8 & -7 & -6 & -5 & -4 & -3 & -2 & -1 & 0 & 1 & 2 & 3 & 4 & 5 & 6 & 7 & 8 & ... \\ 
			\hline 
		Tribonacci Seq. & ...& -8 & 4 & 1 & -3 & 2 & 0 & -1 & 1 & \textbf{0} & \textbf{0} & \textbf{1} & 1 & 2 & 4 & 7 & 13 & 24 & ... \\ 
			\hline
	\end{tabular}}\vspace{.3cm}
	\\Now, consider the plaintext \textbf{P = HELLO2019}. The plain text is divided into blocks as follows:
	\\$P_{1}=[~H~E ~L]=[07~~04~~11], P_{2}=[L ~O ~2]=[11~14~28]~~\text{and}~~P_{3}=[0~1~9]= [26~27~35]$.
	\\\textbf{Encryption:}
	$C \leftarrow  (PK + B)\pmod{37}$.
	\\ $C_{1} = (P_{1}K+B)\equiv$ 
	$\left(\begin{bmatrix}
		07 & 04 & 11 \\
	\end{bmatrix}	
	\\\begin{bmatrix}
		7 & 6 & 4 \\	
		4 & 3 & 2 \\	
		2 & 2 & 1 \\	
	\end{bmatrix}+
	\begin{bmatrix}
		31 & 13 & 19 \\
	\end{bmatrix}\right)\pmod{37}$  \\\indent$\equiv (07~~15~~29) \sim$ (H P 3)
	\\ $C_{2} = (P_{2}K+B)\equiv$ 
	$\left(\begin{bmatrix}
		11 & 14 & 28 \\
	\end{bmatrix}	
	\\\begin{bmatrix}
		7 & 6 & 4 \\	
		4 & 3 & 2 \\	
		2 & 2 & 1 \\	
	\end{bmatrix}+
	\begin{bmatrix}
		31 & 13 & 19 \\
	\end{bmatrix}\right)\pmod{37}$		\\\indent$\equiv (35~~29~~08) \sim$ (9 3 I)
	\\ $C_{3} = (P_{3}K+B)\equiv$ 
	$\left(\begin{bmatrix}
		26 & 27 & 35 \\
	\end{bmatrix}	
	\\\begin{bmatrix}
		7 & 6 & 4 \\	
		4 & 3 & 2 \\	
		2 & 2 & 1 \\	
	\end{bmatrix}+
	\begin{bmatrix}
		31 & 13 & 19 \\
	\end{bmatrix}\right)\pmod{37}$		\\\indent$\equiv (21~~24~~27) \sim$ (V Y 1)
	\\which gives cipher-text $C=(C_{1}C_{2}C_{3})=(HP393IVY1)$
	\\i.e $P : HELLO2019 \rightarrow C : HP393IVY1$
	\\Now, Alice send this cipher-text to Bob along with her signature.
	\\\textbf{Decryption:} After receiving ciphertext $C$ along with signature $(k,B)$, Bob will calculate decryption key $K^{*}$ with the help of their secret key $D$, which is given as
	\begin{eqnarray}
		\lambda = k^{D}\pmod{37}= 4^{13}\pmod{37} \equiv 3 \nonumber
	\end{eqnarray}
	thus
	\\$K^{*} = Q_{\lambda}^{-k}= Q_{3}^{-4} = 
	\begin{bmatrix}
		f_{-2} & f_{-3}+f_{-4} & f_{-3} \\			
		f_{-3} & f_{-4}+f_{-5} & f_{-4} \\		
		f_{-4} & f_{-5}+f_{-6} & f_{-5} \\		
	\end{bmatrix} =
	\begin{bmatrix}
		-1 & 2 & 0 \\		
		0 & -1 & 2 \\		
		2 & -2 & -3 \\		
	\end{bmatrix}$, 
	\\And, $K^{*}\pmod{37} =
	\begin{bmatrix}
		36 & 2 & 0 \\		
		0 & 36 & 2 \\		
		2 & 35 & 34 \\		
	\end{bmatrix}$ \vspace{.2cm}
	\\Clearly, $ K.K^{*} = I\pmod{37}$
	\\Hence, decryption takes place as $P_{i}\leftarrow (C_{i} - B).K^{*}\pmod{37}$.
	\\$P_{1} = (C_{1}-B)K^{*} \equiv$ 
	$\left(\begin{bmatrix}
		07 & 15 & 29 \\
	\end{bmatrix}	
	-
	\begin{bmatrix}
		31 & 13 & 19 \\
	\end{bmatrix}\right)
	\begin{bmatrix}
		36 & 2 & 0 \\		
		0 & 36 & 2 \\		
		2 & 35 & 34 \\		
	\end{bmatrix}\pmod{37} \\ \indent \equiv (07~~04~~11) \sim$ (H E L)
	\\
	$P_{2} = (C_{2}-B)K^{*} \equiv$ 
	$\left(\begin{bmatrix}
		31 & 25 & 06 \\
	\end{bmatrix}	
	-
	\begin{bmatrix}
		31 & 13 & 19 \\
	\end{bmatrix}\right)
	\begin{bmatrix}
		36 & 2 & 0 \\		
		0 & 36 & 2 \\		
		2 & 35 & 34 \\	
	\end{bmatrix}\pmod{37} \\ \indent \equiv (11~~14~~28) \sim$ (L O 2)
	\\$P_{3} = (C_{3}-B)K^{*} \equiv$ 
	$\left(\begin{bmatrix}
		21 & 24 & 27 \\
	\end{bmatrix}	
	-
	\begin{bmatrix}
		31 & 13 & 19 \\
	\end{bmatrix}\right)
	\begin{bmatrix}
		36 & 2 & 0 \\		
		0 & 36 & 2 \\		
		2 & 35 & 34 \\	
	\end{bmatrix}\pmod{37} \\ \indent \equiv (26~~27~~35) \sim$ (0 1 9)
	\\Thus, Bob recovered the plaintext \textbf{HELLO2019} sent by Alice successfully.
	\end{proof}
	\section{Complexity analysis and strength}
	\subsection{Security Strength and Analysis} 
	\begin{definition}[Brute force attack:]
	A brute force attack is a technique of breaking a cryptographic scheme by trying a large number of possibilities(i.e all possible combination). In most of cases, a brute force attack typically means, testing all possible keys in order to recover the information(plaintext) used to produce a particular ciphertext.
	\end{definition}
	One way for an adversary(say Oscar) to break our proposed technique using brute force attack, is to generate all possible matrices. Since, we are working on $F_{p}(p>26$ be a prime), so Oscar need to check $p^{\lambda^{2}}$ matrices.
	\begin{example}
		Let $p=37$ and $\lambda= 50$, then Oscar need to check $p^{\lambda^{2}} = 37^{50^{2}} = 37^{2500}$ matrices, which is equivalent to $3.1938180242\times 10^{3920}$ and too large.
	\end{example}
	As we know that $GL_{\lambda}$ denotes the General Linear group\cite{dummit2004abstract,grillet2007abstract}, which consists of all invertible matrices of order $\lambda\times \lambda$ over $F_{p}$ and its order is given by
	\begin{equation}\label{GLN}
		|GL_{\lambda}(F_{p})| = (p^{\lambda}-p^{\lambda-1})(p^{\lambda}-p^{\lambda-2}) \cdots (p^{\lambda}-p^{1})(p^{\lambda}-1)
	\end{equation}
	\begin{example}
		Consider key matrix $K=Q_{50}^{19}$ ~over $F_{37}$. Here, we have $\lambda=50$ and $p=37$
	\end{example}
	So, total number of invertible matrices of order $50\times 50$ over $F_{37}$ is,
	\begin{eqnarray}
	|GL_{50}(F_{37})| &=& (37^{50}-37^{49})(37^{50}-37^{48}) \cdots (37^{50}-37)(37^{50}-1)\\
	&=& 3.105165707300569\times 10^{3920}\nonumber
	\end{eqnarray}
	Thus in this case, for Oscar to recover this key matrix need to check $10^{3920}$ possible matrices, which is almost difficult.
	Following list shows the strength of system with proposed key matrix $Q_{\lambda}^k$ over $F_{37}$.
	\vspace{.41cm}
	\\\begin{center}
		\begin{tabular}{|c|c|c|c|}
		\hline 
		$\lambda$ & k & Possible no. of matrix to check over $F_{p=37}$ by Oscar\\
		&  & $ |GL_{\lambda}(F_{p})|= (p^{\lambda}-p^{\lambda-1})(p^{\lambda}-p^{\lambda-2}) \cdots(p^{\lambda}-1)$\\  
		\hline
		1 & 1 & $Q_{1}^1 \to |GL_{2}(F_{37})|=36$\\ 
		& 2 & $Q_{1}^2 \to |GL_{2}(F_{37})|=36$\\
		& \vdots & $\vdots$\\
		& k & $Q_{1}^k \to |GL_{2}(F_{37})|=36$\\
		\hline
		 2 & 1 & $Q_{2}^1 \to |GL_{2}(F_{37})|=1.82218\times10^6$\\ 
		 & 2 & $Q_{2}^2 \to |GL_{2}(F_{37})|=1.82218\times10^6$\\
		 & \vdots & $\vdots$\\
		 & k & $Q_{2}^k \to |GL_{2}(F_{37})|=1.82218\times10^6$\\
		\hline
		3 & 1 & $Q_{3}^1 \to |GL_{3}(F_{37})|= 1.26354\times10^{14}$\\ 
		& 2 & $Q_{3}^2 \to |GL_{3}(F_{37})|= 1.26354\times10^{14}$\\
		& \vdots & $\vdots$\\
		& k & $Q_{3}^k \to |GL_{3}(F_{37})|= 1.26354\times10^{14}$\\
		\hline
		4 & 1 & $Q_{4}^1 \to |GL_{4}(F_{37})|= 1.19951\times 10^{25}$\\ 
		& 2 & $Q_{4}^2 \to |GL_{4}(F_{37})|= 1.19951\times 10^{25}$\\
		& \vdots & $\vdots$\\
		& k & $Q_{4}^k \to |GL_{4}(F_{37})|= 1.19951\times 10^{25}$\\
		\hline		
		\vdots & \vdots & \vdots\\ 
		\hline
		$\lambda$ & k & $Q_{\lambda}^k \to |GL_{\lambda}(F_{37})|= (37^{\lambda}-37^{\lambda-1})\cdots (37^{\lambda}-1)$\\ 
		\hline
	\end{tabular}
	\end{center} 
	From above table, it is clear that by making prime $p$ too large, it is almost unbreakable. Also, we observed that it does not depends on $k$ and if we increase the size of key matrix for fixed $p$, then $|GL_{\lambda}(F_{p})| \to \infty$. 
	Thus, if the key space is large, then breaking of system by Brute-force attack\cite{stallings2017cryptography} is impractical. In that case, the opponent may try analysis of ciphertext by various statistical test on it.
	\subsection{Complexity Analysis}
	A encryption scheme's strength is determined by the computational power needed to
	break it. The computational complexity of an algorithm is measured by two factors; time complexity(T) and space complexity(S).	Both T and S are commonly expressed as functions of $n$, where $n$ is the size of the input.
	The time complexity\cite{stothers2010complexity} of an algorithm describes the amount of time it takes to run.
	 In general, the computational complexity of an algorithm is expressed by “big $\mathcal{O}$” notation.\\
	 \textbf{$\mathcal{O}$-notation :} For a given function $g(n)$, $\mathcal{O}(g(n))$ represents the set of functions, 
	 $\mathcal{O}(g(n)) = \{f(n) : \exists$ positive constants $a$ and $n_{0}$ such that $0 \leq f(n) \leq a.g(n)$ for all n,  $n \geq n_{0}$\}
	 
	It is well known that for matrix multiplication, complexity in worst case is $\mathcal{O}(n^{3})$, which increases the time complexity of matrix multiplications of large sizes. But, here in case of generalized Fibonacci matrices this time complexity reduces to $\mathcal{O}(n)$ (See,\ref{GenFeb}).
	\subsection{Conclusion}
	As we know that in a Ring, all non-zero elements does not have necessarily multiplicative inverse. So for encryption key matrix, we use to take entries from field.
	
	But here, in case of multinacci matrices, we do not need to worry about field, it works fine with ring $Z_{n}, n\in \mathbb{N}$. Because as we proved above, in case of multinacci matrices $Q_{\lambda}^{k}$, we have always a matrix $Q_{\lambda}^{-k}$ such that $Q_{\lambda}^{k}Q_{\lambda}^{-k}=I$. Thus multinacci matrices as key matrix for encryption system, play a important role, which increases the key space and complexity of breaking.
	
	In this paper, we have worked on generalization of Fibonacci matrices(corresponding to Fibonacci sequences) and their inverse with direct calculation of $n^{th}$ multinacci term. Also developed a public key cryptography using Affine-Hill cipher with extended Fibonacci(multinacci matrix) matrix $Q_{\lambda}^{k}$ under prime modulo, where $Q_{\lambda}^{k}$ is matrix of order ${\lambda}\times\lambda$ and its elements are constructed using terms of multinacci sequences.
	
	Our proposed method strengthen the security of the system which has three digital signatures namely $\lambda$, $ k $ and $ B $. Since, $ K $(Corresponding to $k$, $K= Q_{\lambda}^{k}$) and B are known only to Alice and Bob, So it is not possible to break this system by anyone. Our proposed method is mathematically simple and having large key space, as for construction of key for known party is easy but for intruder it is very difficult to construct a matrix with the help of tuple $(\lambda, k)$. This is the main beauty of our proposed schemes for key establishment. This proposed system take care of authenticity and integrity of data as key k and $\lambda$ is known only to Bob and Alice.
	\subsection*{Acknowledgment}
	First author acknowledge the University Grant Commission, India for providing fellowship for this research work and also thank to Central University of Jharkhand, India for their kind support.

\bibliography{refhillcipher}
\bibliographystyle{acm}
\end{document}